\RequirePackage{fix-cm}
\documentclass[smallextended]{svjour3}       
\smartqed
\usepackage{graphicx}
\usepackage{theorem,latexsym,amssymb}
\usepackage[small]{caption}
\usepackage{graphicx}
\usepackage{amsmath}
\usepackage{booktabs}
\usepackage{natbib}
\usepackage{color}
\usepackage{url}

\usepackage[moderate]{savetrees}

\usepackage[boxed]{algorithm}
\usepackage{algorithmic}
\usepackage{eqparbox}

\newcommand{\midd}{\mathbin{:}}

\begin{document}




\title{Cutoff stability under distributional constraints with an application to summer internship matching\footnote{This paper is the extension of our conference paper~\citep{aziz2020summer}. The main addition is Section \ref{sec:cutoff} on our new solution concept of cutoff stability and an algorithm to achieve it.}}

\author{
	Haris Aziz \and
	Anton Baychkov \and
	P\'eter Bir\'o
}

\titlerunning{}
\authorrunning{Haris Aziz, Anton Baychkov, and P\'eter Bir\'o}

\institute{Haris Aziz \at
              UNSW Sydney, Australia \\
              \email{haziz@cse.unsw.edu.au}
              \and
			  Anton Baychkov \at
              University of Sydney, Australia \\
              \email{abay3963@uni.sydney.edu.au}
        \and
	  P\'eter Bir\'o \at
        Centre for Economic and Regional Studies, Hungary \\
        \email{peter.biro@krtk.mta.hu}	
}
\date{\vspace{-2em}Received: date / Accepted: date}

\maketitle


%
%
%

\begin{abstract}
We introduce a new two-sided stable matching problem that describes the summer internship matching practice of an Australian university. The model is a case between two models of Kamada and Kojima on matchings with distributional constraints. We study three solution concepts, the strong and weak stability concepts proposed by Kamada and Kojima, and a new one in between the two, called cutoff stability. Kamada and Kojima showed that a strongly stable matching may not exist in their most restricted model with disjoint regional quotas. Our first result is that checking its existence is NP-hard.
We then show that a cutoff stable matching exists not just for the summer internship problem but also for the general matching model with arbitrary heredity constraints. We present an algorithm to compute a cutoff stable matching and show that it runs in polynomial time in our special case of summer internship model.
However, we also show that finding a maximum size cutoff stable matching is NP-hard, but we provide a Mixed Integer Linear Program formulation for this optimisation problem.
\end{abstract}

\vspace{-1em}
\section{Introduction}

\sloppy

Centralized two-sided matching market algorithms have received immense success in several application domains, including matching students to schools, residents to hospitals, and projects to workers.\footnote{For an overview of real-life matching markets, please see \url{http://www.matching-in-practice.eu}, and a recent survey~\citep{biro2017applications}.} We present a novel matching market model that we refer to as the summer internship matching market. The model captures the matching of student applicants to projects proposed by supervisors in an internship program.
A distinctive feature of the model is that in order for an applicant to be assigned to any project, a certain amount of money needs to be contributed from the project supervisors' funds.


Our problem is inspired by summer intern research programs in Australia.
It is common for undergraduate students to undertake research projects over the summer. Each project is supervised by one or more members of the faculty, with many offering multiple projects. Even though the projects may be discounted by contributions from the faculty, supervisors are often required to contribute to the funding of these positions, from their personal research budget. Alternatively, they could be constrained by the amount of time they can allocate to supervision. These supervisor-side constraints mean that not all projects can be funded.

Just as the standard hospital resident matching models do not just apply to the matching of doctors to hospitals~\citep{Roth08a}, our model also does not just apply to matching of student interns. It applies to any two-sided matching model in which very widely applicable budget requirements and budget constraints are involved. For example, our problem also models hiring scenario in which different teams have their own budgets and they want to hire employees. Certain employee roles could sit across various teams. In that case, multiple teams can pool in their money to fund joint positions.

The budget constraints that we consider lead to interesting research challenges.
Firstly, applying standard matching algorithms such as the Deferred Acceptance Algorithm does not work as it do not deal with complex feasibility constraints. More critically, as we will discuss, there may exist no feasible matching that satisfies stability as considered in seminal papers on matching (see, e.g., \citet{GaSh62a} and \citet{AbSo03b}).



\vspace{-1em}
\subsection{Contributions}

In this paper, we formalize the summer internship problem with budgets, abbreviated as SIP. It falls within the class of models that matches applicants, projects, and supervisors. The main characteristic of our problem is that supervisors have budgets that they can spread across their projects.
Our model is more general than the widely-used hospital-resident matching model in which the hospitals are partitioned into regions and regions have upper capacities~\citep{KaKo15}, a model that we abbreviate as REG. On the other hand, our model is a special case of the matching model under distributional constraints of \citet{KaKo17b}, where the feasibility of a matching is monotone in the number of applicants matched, the property called \emph{heredity} by~\citet{GKK+17}, that we abbreviate as HER.

First, we study the concept of strong stability proposed in~\citet{KaKo15}, where the authors showed that such solution may not exist for REG. Here we prove that the problem of checking the existence of a strongly stable matching is NP-hard for REG (and thus also for any settings that contain REG as a special case, such as SIP and HER).

Then we study weak stability, also introduced in~\citet{KaKo15} for REG and then studied in~\citet{KaKo17b} for HER. In the conference version of our paper \citep{aziz2020summer} we provided a strongly polynomial algorithm for computing a weakly stable matching for our intermediate SIP model, based on the algorithm of~\citet{KaKo17b} for HER. Our algorithm uses as an oracle an algorithm based on network flows to repeatedly check whether a given matching is feasible or not.
In this extended version of our conference paper, we strengthen this result by presenting a general algorithm that returns a matching satisfying cutoff stability.
Fair matchings are exactly those matchings that can be induced by a set of cutoff scores. Meanwhile, cutoff stable matchings are those that are induced by minimal cutoffs, i.e., where the decrease of any cutoff would make the induced matching infeasible. We show that cutoff stability is an intermediate notion between weak and strong stability, and the computation of a cutoff stable matching is always possible with our algorithm for HER. Since HER covers many well-studied settings including refugee matching (see, e.g., \citep{ACGS18a,DKT17a}), our result have wide-applicability.\footnote{See, for example, the discussion by \citet{KaKo20a} who point out that even intra-project heredity constraints capture problems including college admissions with students with disabilities; refugee match, and daycare allocation.} We apply the algorithm in the context of SIP and show that it can be implemented to run in polynomial time. However, we also show that finding a maximum size cutoff stable matching is NP-hard even for REG (implying that is also NP-hard for SIP and HER). However, we formulate a Mixed Integer Linear Program (MILP) for finding such an optimal solution for SIP. Many of our general structural results for the HER model. Unless specified otherwise, we will assume that our formal statements apply to HER.

Finally, we provide a normative criterion for egalitarian ex-post allocation of supervisor funding among projects, and a polynomial-time algorithm to find an egalitarian allocation. Combined with our polynomial-time algorithm for finding a weakly stable matching, we present a compelling approach for finding a desirable solution for the summer internship problem that appeals to both stability and fairness requirements.
\vspace{-0.5em}
\subsection{Related Work}

The literature on two-sided matching was inspired by the seminal paper of \citet{GaSh62} who considered matching markets that match students to schools and hospitals to residents.
The paper has spawned richer matching models and resulted in new algorithmic work~(see, e.g., \citet{Manl13a}). Our work is an extension of these models and falls in the general umbrella of matching markets with various kinds of distributions constraints (see, e.g., \citet{ABY22a,ACGS18a,KaKo17a,FIT+16a,FrTr17,KHIY17a}).

Our concept of budget-feasibility is a type of feasibility constraint, as defined by \citet{KaKo17b} and thus our problem falls under the umbrella of matching under feasibility constraints. Thus, the notions of strong and weak stability studied by \citet{KaKo17b} also apply in our model.
In our paper, we focus on computational results such as establishing NP-completeness or polynomial-time solvability of stable matchings. Our model also has the additional dimension of budget allocations for which we explore fairness concepts as well as algorithms to divide the budget in an egalitarian manner.
We also propose an intermediate concept called cutoff stability and prove that it can be achieved for the general matching model with arbitrary heredity constraints. Cutoff stability is a new notion for the HER model, though similar notions have been studied for different stable matching models, such as the choice function model of ~\citet{FleinerJanko2014}. Our concept of cutoff stability defined for general feasibility constraints is similar in spirit to  the within-type envy-freeness concept of \citet{EcYe15a} in the context of school choice with affirmative action.

\citet{GKK+17} also considered two-sided matching under general feasibility constraints that satisfy the heredity property. Their main contribution is proposing an algorithm called Adaptive Deferred Acceptance that satisfies strategy-proofness, non-wastefulness, and a fairness property (that is weaker than the weak stability concept of \citet{KaKo17b} and hence also cutoff stability). \citet{KaKo20a} also consider heredity constraints that apply to individual schools/hospitals.


Our model bears some similarities with the hospital-resident matching problem with regional constraints~\citep{KaKo15,BFIM10,GIKY+16a,AGSW19a,KaKo18a}. In these region-based problems, at most a certain number of students can be selected from given regions. On the other hand, in the summer internship problem, a supervisor's budget is divisible and can be spread partially over all of her projects.
If the regions are disjoint (as studied by \citet{KaKo15}), then the region-based model is a special case of our model.
The general setting with region constraints has not seen many positive results, and the more constrained hierarchical regions as studied by \citet{GIKY+16a} and \citet{KaKo18a} can neither replicate, nor be replicated by a set of supervisors. In particular, the concept of stability with regional priorities proposed by \citet{KaKo18a} cannot be directly applied to our setting. Our concept of cutoff stability is an alternative intermediate notion of stability and it can be applied to a wider range of settings (those subsumed under the HER model).


\citet{AIM07} considered a different model for student-project allocation. Notable differences in their model include: (1) no project has multiple supervisors, (2) each supervisor has a universal priority list over students which is not project specific, and (3) a supervisor has a rigid capacity constraint for the number of projects to supervise. Our model allows supervisors to explore more efficient outcomes by pooling in their budgets to host a student. The universal priority list of each supervisor makes the model of \citet{AIM07} much more restricted and different from our model.
%

Two other recent models are similar to our setting. \citet{GKK+17} introduce the Student-Project-Room matching problem. Rooms are indivisible, and at most one room can be allocated to each project. \citet{IYY18} extend this to a more general Student-Project-Resource allocation problem. Resources are still indivisible, but there is no longer a restriction imposed on the number of resources that can be allocated to each project. Our model is distinct from both of these, as it allows the resources (in our case, supervisor budgets) to be divisible. This divisibility allows for better computational results. For instance, verifying the feasibility of a matching, and finding a weakly stable (and thus non-wasteful) matching can be done in polynomial time.

There is also work on matching with budget constraints~(see, e.g., \citet{KaIw17a,KaIw18a,IHZ+19a}). The models considered in these papers are different in several respects. For example, hospitals have additive utilities and each hospital gives monetary compensation to doctors.

Regarding the LP descriptions and IP techniques for two-sided stable matching problems, \cite{baiou2000stable} gave the first description on the stable admissions polytope. Integer programming techniques have been used later for college admissions with special features \citep{ABMcB2016}, stable project allocation under distributional constraints \citep{ABSz2018}, the hospital--resident problem with couples \citep{BMMcB2014}, and ties \citep{KM2014,Delormeetal2019}.

\vspace{-1em}
\section{Preliminaries}

First, we introduce our summer internship matching model (SIP), we show that checking the feasibility of a matching is polynomial time tractable, and finally we also show the relation of our model to the REG and HER models of Kamada and Kojima described in \citet{KaKo15} and \citet{KaKo17b}, respectively.
\vspace{-0.5em}
\subsection{Model}

Let $A$ be a finite set of applicants, and $P$ a finite set of projects. Each applicant $a\in A$ has a strict preference list $\succ_a$ that ranks the projects that $a$ finds acceptable. Each project $p\in P$ has a preference list $\succ_p$ over the subset of applicants that $p$ finds acceptable, and a maximum capacity $c_p$.

Furthermore, let $S$ denote the set of project supervisors. Each supervisor $s\in S$ has a list of projects $P_s$ that they supervise, and is endowed with a budget (e.g. quantity of funds) $q_s$ that they can allocate among those projects. We assume that these budgets are infinitely divisible and that each applicant requires one unit of funding. Further, we assume that these endowments are publicly known, and thus supervisors cannot strategise by misreporting their budgets.
Additionally, denote the list of supervisors for project $p$ by $S_p$.

We say that an applicant $a\in A$ is matched to project $p\in P$ if $(a, p)\in M$. A \textbf{matching} $M$ is a subset of $A\times P$ that satisfies the following conditions:

\begin{itemize}
	\item Each applicant is matched to at most one project (for all $a\in A$, $|\{(a, p)\in M\ \midd \ p\in P\}|\leq 1$), and $a$ finds the project they are matched to acceptable.
	\item The number of applicants matched to any project does not exceed that project's capacity (for all $p\in P$, $|\{(a, p)\in M\ \midd \ a\in A\}|\leq c_p$), and $p$ finds all applicants matched to it acceptable.
\end{itemize}

We use $M(a)$ to refer to the project that applicant $a$ is matched to ($M(a)=\emptyset$ if $a$ is unmatched). Meanwhile, $M(p)$ denotes the set of applicants matched to $p$.

Let $x_{s,p}$ be the amount of funds a supervisor $s$ allocates to project $p$. We call a matching $M$ \textbf{feasible} (or supervisor-feasible) if there exists a set $\{x_{s,p}\}_{s\in S, p \in P_s}$ that satisfies the following conditions:

\begin{itemize}
	\item $x_{s,p}\geq 0$ for all $s\in S,~p\in P_s$
	\item Every project receives one unit of funding for each applicant matched to it: $\sum_{s\in S_p} x_{s,p} = |M(p)| \text{ for all } p\in P$
		\item Supervisors do not exceed their endowment: $\sum_{p\in P_s} x_{s,p} \leq q_{s} \text{ for all } s\in S$
\end{itemize}


We call any set $\{x_{s,p}\}_{s\in S, p \in P_s}$ that satisfies the above conditions for matching M a \textbf{feasible funding allocation}.
The mathematical model presented exactly captures the student research internship program in our university: each supervisor can be part of multiple internship project proposals but does not necessarily have the funding to contribute to all of them.

\begin{example}[Summer Internship Problem]
	Consider the following instance of the summer internship problem with 2 applicants 2 supervisors, and 2 projects.
	\begin{align*}
		A=\{a_1,a_2\}&&
		\succ_{a_1}:p_2, p_1&&
		\succ_{a_2}:p_1, p_2\\
		P=\{p_1,p_2\}&&
		\succ_{p_1}:a_1, a_2&&
		\succ_{p_2}:a_2, a_1\\
		S=\{s_1, s_2\}&&
		P_{s_1}=\{p_1, p_2\}&&P_{s_2}=\{p_2\}\\
		q_{s_1}=0.7&&
		q_{s_2}=0.5&&
		c_{p_1}=c_{p_2}=1
	\end{align*}
	
	The only three feasible matchings are the empty matching and the two matchings in which some applicant is matched to project $p_2$. The reason no one can be matched to project $p_1$ is that $p_1$ has a sole supervisor $s_1$ who does not have sufficient funding to fund $p_1$. On the other hand, the combined funding of $s_1$ and $s_2$ is more than 1 for project $p_2$ so $p_2$ can be funded. \textcolor{black}{One possible} feasible funding allocation $x$ is where $s_1$ contributed half of the budget of project $p_2$ and $s_2$ contributes the rest: $x_{s_1,p_2}=0.5$, $x_{s_2,p_2}=0.5$.
	\end{example}
\vspace{-0.5em}
\subsection{Checking Feasibility of a Matching}

We show that checking feasibility of a matching in our model can be done efficiently by reducing the question to a network flow problem.\footnote{For an overview of network flows, see \citet{AMO93a}.} Define the \textbf{funding flow graph} $G_M$ associated with a matching $M$ as follows:

\begin{itemize}
	\item $V(G_M) =\{s^*\}\cup S\cup P\cup \{t^*\}$, where $s^*$ is the source, and $t^*$ is the sink
	\item Arcs $(s,p)$, for all supervisor-project pairs where $p\in P_s$ with capacity $\infty$
	\item Arcs $(s^*,s)$, for all $s\in S$, each with capacity $q_s$
	\item Arcs $(p,t^*)$, for all $p\in P$, each with capacity $|M(p)|$
\end{itemize}

\begin{theorem}\label{thm:feasibilitycheck}
	The feasibility of a matching can be checked in polynomial time $O(\max\{|S|,|P|\}^3)$ for the summer-internship problem.
\end{theorem}	
	
\begin{proof}
	Our first claim is that a matching $M$ is feasible if and only if $G_M$ admits a feasible $s^*$-$t^*$ flow of size $|M|$. \textcolor{black}{Note that by construction the value of an $s^*$-$t^*$ flow is at most $|M|$.}
	
	Suppose $M$ is feasible. Define a flow $f$ on $G_M$ as follows:
	\begin{itemize}
		\item $f(s,p)=x_{s,p}$, $\forall s\in S, p\in P_s$
		\item $f(s^*,s)=\sum_{p\in P_s} x_{s,p}$, $\forall s\in S$
		\item $f(p,t^*)=\sum_{s\in S_p} x_{s,p}$, $\forall p\in P$
	\end{itemize}
It is easy to see that this is a feasible flow of size $|M|$.
Now suppose that $G_M$ admits a feasible flow $f$ of size $|M|$.
	Set $x_{s,p}=f(s,p)$, $\forall s\in S, p\in P_s$. We can then show that this $\{x_{s,p}\}$ satisfies the conditions of feasibility.

Now that we have established 	the claim, we use the fact that
the maximum flow problem can be solved in $O(|V|^3)$ time (where $V$ is the set of vertices), using for instance, the algorithm proposed by \citet{MKM78}. $V(G_M)=|S|+|P|+2$ and, given a matching $M$, $G_M$ can be constructed in $O((|S|+|P|)^2)$ time. We can check whether $M$ is feasible in $O(\max\{|S|,|P|\}^3)$ time by computing a maximum flow and verifying whether it equals $|M|$.
\qed\end{proof}

Note that by the integer property of the network flow problem, if all the capacities of the supervisors are integer and the flow is feasible then an integer funding allocation exists.
\vspace{-0.5em}
\subsection{Connection with the models by Kamada and Kojima}
	
Our model satisfies the heredity property of \citet{KaKo17b} and \citet{GKK+17}, that says that if a matching is feasible then it remains feasible if the numbers of applicants matched to each project \textcolor{black}{decreases or remains the same}. \textcolor{black}{Let $\mathcal{M}$ denote all the subsets of $A\times P$, that we call as generalized matchings.} Their matching model under distributional constraints that satisfy heredity (HER), can be represented by a \textbf{feasibility function} $f:\mathcal{M}\rightarrow \{0,1\}$ that satisfies the following condition:

\begin{itemize}
\item $f(\emptyset)=1$ and, for any two matchings $M$ and $M'$, if $|M'(p)|\leq |M(p)| \ \forall p\in P$ then $f(M)=1 \Rightarrow f(M')=1$.
\end{itemize}

A \textcolor{black}{generalized} matching $M$ is feasible if and only if $f(M)=1$. Note that the feasibility constraints are `anonymous' in the sense that they do not depend on the identity of the applicants matched but only on their quantity. \textcolor{black}{When considering the feasibility of a generalized matching in the context of SIP we relax the condition that every applicant can be matched to at most one project, but we obey the further requirements, namely the project-capacity and budget conditions. These conditions can be still checked efficiently by the very same reduction to the maximum flow problem.}

\citet{KaKo15} studied a basic model with regional upper quotas, where the regions form a partition of the set of hospitals (REG). We can represent this in our model by replacing each region with a supervisor, endowed with a budget equal to that region's capacity, and supervising a set of projects that correspond to the hospitals in that region.


Thus, REG is a special case of SIP, and SIP is a special case of HER. Therefore, every hardness result for REG implies the same hardness result for SIP and HER, and any easiness result for HER implies the same result for SIP, and an easiness result for SIP also holds for REG.
\subsection{Fair matchings and cutoff scores}

We say that matching $M$ is \emph{fair}, if for every pair $(a, p)\notin M$, $p\succ_aM(a)$ implies that $a' \succ_p a$ for every $a'\in M(p)$.\footnote{In the school choice literature this property is also called as \emph{justified envy-freeness}, see, e.g., \citet{abdulkadirouglu2003school}.} This is a basic property for both strong and weak stability by ~\citet{KaKo17b}, and so also for our intermediate notion of cutoff stability.

In the classical Gale-Shapley model \citep{GaSh62} a matching is stable if and only if it is fair and non-wasteful, where \emph{non-wastefulness} means that there exists no project $p$ with unfilled capacity and applicant $a$ where $p\succ_aM(a)$. In our context we can define different non-wastefulness notions leading to weak and strong stability and a new intermediate stability concept, called cutoff stability. We will define and analyze these stability concepts in Sections \ref{sec:strong}, \ref{sec:weak} and \ref{sec:cutoff}, respectively.

The concept of cutoff scores is closely related to fair matchings. Let $d:P\rightarrow [0,1,\dots,|A|+1]$ be the cutoff score function, where $d(p)$ is the \emph{cutoff} at project $p$. Without loss of generality we assume that each project $p$ assigns a score to each applicant $a$ in accordance with its preference list, that is $a$ has score $|A|-k+1$ if she is ranked $k$th by project $p$. Given cutoff scores $d$, we say that applicant $a$ is \emph{admissible} to project $p$ if her score achieves the cutoff. Cutoff scores $d$ \emph{induce} matching $M$, if every applicant is matched to the best project of her preference where she is admissible.\footnote{In fact, in many college admission schemes only the cutoff scores are announced, the induced matchings are obvious for the participants involved, see, e.g. the cases of Hungary~\citep{ABMcB2016} and Australia~\citep{artemov2017strategic,guillen2020field}.} The following observation is well-known (see, e.g. Lemma 3 in \citet{FleinerJanko2014}), for completeness we give a short proof.

\begin{proposition}\label{prop:fair}
A matching is fair if and only if it is induced by some cutoff scores.
\end{proposition}
\begin{proof}
A matching induced by cutoffs is always fair by definition, since $p\succ_aM(a)$ implies that $a$ could not reach the cutoff score at $p$ so everyone assigned to $p$ has a higher rank than her. In the other direction let $M$ be a fair matching. For every project $p$ we set the cutoff $d(p)$ to be equal to the score of the lowest ranked applicant in $M(p)$. It is easy to see that $d$ induces $M$.
\qed\end{proof}

In the following sections we will consider the strong and weak stability concepts of \citet{KaKo17b}, and subsequently introduce our new solution concept of cutoff stability.
	
\vspace{-1em}
\section{Strong Stability}\label{sec:strong}

An applicant-project pair $(a, p)$ is a \textbf{blocking pair} for matching $M$ if:
\begin{itemize}
	\item applicant $a$ prefers project $p$ to the project they are currently matched to: $p \succ_a M(a)$, and
	\item either:
	\begin{itemize}
		 \item $p$ is under capacity and finds $a$ acceptable: $|M(p)|<c_p$ and $a \succ_p \emptyset$; or
		 \item $p$ prefers $a$ to one of its currently matched applicants: $\exists \  a'\in M(p)$ such that  $a \succ_p a'$
	\end{itemize}
\end{itemize}

\begin{definition}[Strong Stability]
We call a matching $M$ strongly stable if for any blocking pair $(a, p)$ for matching $M$, the following two conditions are satisfied:
\begin{itemize}
	\item $a'\succ_{p}a$ for all applicants $a'\in M(p)$
	\item The matching $M'=(M\cup \{(a, p)\}) \setminus \{(a,M(a))\}$ is not feasible.
\end{itemize}
\end{definition}

The first condition implies that we only allow the existence of blocking pairs involving a project $p$ that is under its maximum capacity $c_p$. The second condition implies that even if  slot is free at project $p$, adding applicant $a$ to project $p$ will result in a distributional constraint being violated so that $M'=(M\cup \{(a, p)\}) \setminus \{(a,M(a))\}$ is not feasible.
Thus, we allow blocking pairs that cannot be satisfied without violating our feasibility constraint to exist in a strongly stable matching.

Alternatively, we can define strong stability in terms of fairness and strong non-wastefulness, as follows. We say that matching $M$ is \emph{strongly non-wasteful} if for every pair $(a, p)\notin M$, $p\succ_aM(a)$ implies that $(M\cup \{(a, p)\})\setminus \{(a,M(a))\}$ is not feasible or project $p$ is at capacity. It is easy to see that $M$ is strongly stable if and only if it is fair and strongly non-wasteful.


Our first observation is that for our problem, a strongly stable matching does not exist. This follows from the observation that our model is more general than the hospital resident setting with disjoint regions~\citep{KaKo15}. We provide an adaptation of an example (Example 1) {of~\citet{KaKo17b}} for the sake of completeness.

\begin{example}[A strongly stable matching does not necessarily exist for the summer internship problem]\label{ex:unsolvable}
	
	Consider the following instance of the summer internship problem:
	\begin{align*}
		A=\{a_1,a_2\}&&
		\succ_{a_1}:p_2, p_1&&
		\succ_{a_2}:p_1, p_2\\
		P=\{p_1,p_2\}&&
		\succ_{p_1}:a_1, a_2&&
		\succ_{p_2}:a_2, a_1\\
		S=\{s\}&&
		P_s=\{p_1,p_2\}&&
		q_s=c_{p_1}=c_{p_2}=1
	\end{align*}

	It is easy to see that $|M|\leq 1$. If both applicants are unmatched then $(a_1,p_1)$ forms a blocking pair. Suppose without loss of generality that $a_1$ is matched in the \textcolor{black}{feasible} matching. If $M=\{(a_1,p_1)\}$ then $(a_1,p_2)$ is a blocking pair that does not satisfy the second condition of strong stability. If $M=\{(a_1,p_2)\}$ then $(a_2,p_2)$ is a blocking pair that does not satisfy the first condition of strong stability. Thus, every feasible matching admits a blocking pair that is not permitted under strong stability, and is therefore not strongly stable.
	\end{example}

Below we show that the problem of deciding the existence of a strongly stable matching is NP-complete. We reduce for Restricted MAX-SMTI, the problem of deciding whether there exists a complete stable matching for the stable marriage problem with incomplete lists and ties under the restriction that the preferences of the men are strict, and the preference list of each woman is either strict or consists solely of a tie of length two~\citep{MII+02a}.  A matching is said to be weakly stable if it is not blocked by a pair where both parties strictly prefer each other to their current partners.

First, we introduce an instance that will serve as the core of the construction imitating an indifferent woman.

\begin{example}[An instance with two strongly stable matchings, covering different agents]\label{ex:multiple}
\color{black} The following instance is the simplest possible to demonstrate the above mentioned property.
	\begin{align*}
		A=\{a_1,a_2\}&&
		\succ_{a_1}:p_1&&
		\succ_{a_2}:p_2\\
		P=\{p_1,p_2\}&&
		\succ_{p_1}:a_1&&
		\succ_{p_2}:a_2\\
		S=\{s\}&&
		P_s=\{p_1,p_2\}&&
		q_s=c_{p_1}=c_{p_2}=1
	\end{align*}

One can easily check that there are two strongly stable matchings:
$M_1=\{(a_1,p_1)\}$ and $M_2=\{(a_2,p_2)\}$.
\color{black}
\end{example}

Note that this example also shows that the Rural Hospitals' Theorem of  \citet{roth1986allocation} does not hold for strongly stable matchings, which says that always the same applicants are matched in every stable solution in a many-to-one college admissions problem. Actually, the same two matchings are also the only weakly stable matchings, and thus the only cutoff stable matchings, by the same argument.

Using the example as a gadget, we will show that deciding whether an instance of our problem has a strongly stable solution is an NP-complete problem.

\begin{theorem}\label{NP-completeness}
\textcolor{black}{Checking the existence of a strongly stable matching is NP-complete even for REG.}
	
\end{theorem}
\begin{proof}
\color{black}We describe the proof in the context of SIP, where all the supervisors have capacity one and each is responsible for at most two distinct projects. Given a solution, we can check whether it is strongly stable by considering each potential blocking pair in polynomial time, so the problem is in NP. For proving NP-hardness, we reduce again from Restricted MAX-SMTI problem~\citep{MII+02a} in two steps.

\emph{Part i)} Suppose that we have an instance $I$ of MAX-SMTI. First we will create an instance $I'$ of SIP, such that there is a one-to-one correspondence in between the weakly stable matchings in $I$ and the strongly (and cutoff and weakly) stable matchings in $I'$.
For $I$, we use the same notation of $W=W^s\cup W^t$ to denote the set of women with strict preferences and with a single tie, respectively, and let $U$ denote the set of men. Every man in $I$ will be replaced by an applicant in $I'$ with essentially the same preferences. Now, every woman in $w_j\in W^s$ will be replaced with a single project $p_j$ with identical preferences. For every woman $w_j\in W^t$ we create two projects, $p_j^1$ and $p_j^2$ that will correspond to the projects in Example \ref{ex:multiple}. If $u_i$ and $u_k$ were the two men in the single tie of $w_j$ in $I$ then let the two corresponding applicants $a_i$ and $a_k$ complete the instance of Example \ref{ex:multiple}. So the preference lists of the projects are $\succ_{p_j^1}: a_i$ and $\succ_{p_j^2}: a_k$, whilst $a_i$ has only $p_j^1$ in her preference list, and $a_k$ has only $p_j^2$ in her preference list. In the first part of the reduction we show that if $M$ is a weakly stable matching in $I$ then there is a corresponding strongly stable matching $M'$ in $I'$ with the same size, and vice versa. The correspondence between the matchings is as follows.
\begin{itemize}
\item for every $w_j\in W^s$ and $u_i\in U$, $(u_i,w_j)\in M \iff (a_i,p_j)\in M'$
\item for every $w_j\in W^t$ and $u_i\in U$, where $u_i$ is the first man in the tie of $w_j$, $(u_i,w_j)\in M \iff (a_i,p_j^1)\in M'$
\item for every $w_j\in W^t$ and $u_k\in U$, where $u_k$ is the second man in the tie of $w_j$, $(u_k,w_j)\in M \iff (a_k,p_j^2)\in M'$
\end{itemize}
We can observe that the sizes of $M$ and $M'$ are the same.

\emph{Part ii)} In the second part of the reduction we create an extended instance $I''$ of SIP by adding a gadget $G^*$, which is a copy of the unsolvable instance in Example \ref{ex:unsolvable} consisting of two projects $P^*=\{p_1^*,p_2^*\}$ and two applicants $A^*=\{a_1^*,a_2^*\}$, together with an additional project $p^*$. Let suppose that $p^*$ accepts one of the applicants in $G^*$, say $a_1^*$. Let $p^*$ be the most preferred project by $a_1^*$, so including this project in $G^*$ would turn this instance solvable by assigning $p^*$ to $a_1^*$ and $p_1^*$ to $a_2^*$. To link $G^*\cup p^*$ with the rest of $I'$ we put all the applicants in $I'$ ahead of $a_1^*$ in the preference list of $p^*$ in a random order, and we also append $p^*$ to the end of the preference list of each applicant in $I'$.

We will show that $I$ has a complete weakly stable matching if and only if $I''$ has a strongly stable matching. Let us suppose first that $M$ is a complete stable matching of $I$, we construct a strongly stable matching $M''$ of $I''$ as follows. First we create $M'$ in $I'$ that covers all the applicants in $I'$, and then we create $M''$ by adding $\{(a_1^*,p^*),(a_2^*,p_1^*)\}$. In the other direction, if $M''$ is a strongly stable matching in $I''$ then $p^*$ must be assigned to $a_1^*$, as otherwise gadget $G^*$ would not admit a strongly stable solution. Hence all the applicants in $I'$ must be matched to projects in $I'$, since any unmatched applicant would block with $p^*$ otherwise. Therefore, we can create the corresponding complete matching $M$ in $I$ that must be weakly stable due to part i) of the reduction.
\qed\end{proof}

\color{black}
Our hardness result is strong because it holds for a very restricted setting of \citet{KaKo15} that concerns disjoint regions (REG), even if at most two hospitals belong to each region.\footnote{Most of the computational hardness results for distributional constraints concern overlapping regions (see, e.g. \citet{GIKY+16a}).} This case can also occur when one hospital has a common upper quota for two different types of jobs, e.g., daytime and night shifts, or surgical and medical internship positions. Another motivating example is the Hungarian college admission scheme, where students can be admitted to a programme under two contracts, state-funded and privately-funded, and there is a common upper bound on them~\citep{BFIM10}.


\vspace{-1em}
\section{Weak Stability}\label{sec:weak}
	
In view of the non-existence and NP-completeness of checking the existence of strongly stable matchings, one can consider a weaker stability criterion. \citet{KaKo17b} proposed a weak stability concept for a setting that does not concern budgets but which has an abstract feasibility indicator function for any given matching. We present the definition in our terminology of applicants and projects.

\begin{definition}[Weak Stability]
	We call a matching $M$ weakly stable if for any blocking pair $(a, p)$ for matching $M$, the following two conditions are satisfied.
	\begin{itemize}
		\item $a'\succ_{p}a$ for all applicants $a'\in M(p)$
		\item \textcolor{black}{generalized matching} $M\cup \{(a, p)\}$ is not feasible.
	\end{itemize}
\end{definition}

Note the similarity in the definition of strong stability and weak stability. The only difference is that in the second condition, applicant $a$ can have two contracts: one with project $M(a)$ and another with the project $p$ she is blocking with. One way to see this is that in order for applicant $a$ to block with $p$, it must sign the contract with $p$ before it opts to annul its match with project $M(a)$.
We call any blocking pair that satisfies these condition \textbf{permitted under weak stability}. 

Alternatively, we can also define weak stability in terms of fairness and weak non-wastefulness. We say that matching $M$ is \emph{weakly non-wasteful} if for every pair $(a, p)\notin M$, $p \succ_aM(a)$ implies that $M\cup \{(a, p)\}$ is not feasible or project $p$ is at capacity. It is easy to see that $M$ is weakly stable if and only if it is fair and weakly non-wasteful, for a short proof see Proposition 1 in \citet{KaKo17b}.

We say that $p$ is \emph{unconstrained} for a feasible matching $M$ if for any $(a, p)\notin M$, $M\cup\{(a, p)\}$ remains feasible, i.e., when one more applicant can be added to project $p$ by keeping the solution feasible.

\begin{proposition}
A matching $M$ is weakly stable if and only if it is induced by cutoff scores $d$ such that for every unconstrained project $p$, $d(p)=0$.
\end{proposition}

\begin{proof}
Following Proposition \ref{prop:fair}, we only need to show that the condition of every unconstrained project having zero cutoff is equivalent to weak non-wastefulness. Suppose first that $M$ is weakly non-wasteful. This means that there cannot exist an unconstrained project $p$ and an applicant $a$ such that $p\succ_aM(a)$, so indeed for all unconstrained project we can set the cutoff to be zero, and for the constrained projects we just set the cutoff to be equal to the score of lowest ranked applicant assigned. In the other direction, if we have cutoff scores satisfying that every unconstrained project has cutoff score zero, then there cannot exist an applicant $a$ such that $p\succ_aM(a)$ for an unconstrained project $p$.
\qed\end{proof}

In the conference version of our paper~\citep{aziz2020summer} we presented a polynomial-time algorithm that always returns a weakly stable matching based on the algorithm by \citet{KaKo17b} (Appendix B.3). 
In this extended version of our conference paper, we strengthen this result by showing that a so-called cutoff stable matching can also be computed efficiently for SIP, by designing an algorithm that also works for HER.

\vspace{-1em}
\section{Cutoff stability}\label{sec:cutoff}


In this section, we discuss cutoff stability that applies to any matching problem with feasibility constraints.

We say that $M$ is \emph{cutoff non-wasteful} if $M$ is non-wasteful and for every pair $(a, p)\notin M$, $p\succ_a M(a)$ implies that either
\begin{itemize}
\item $(M\cup \{(a, p)\})\setminus \{(a,M(a))\}$ is not feasible, or
\item there exists another applicant $a'\notin M(p)$, such that $a'\succ_p a$, $p\succ_{a'}M(a')$ and $(M\cup \{(a',p)\})\setminus \{(a',M(a'))\}$ is not feasible, or
\item project $p$ is at capacity
\end{itemize}
We say a matching is \emph{cutoff stable} if it is fair and cutoff non-wasteful.
\textcolor{black}{To get an intuition about the meaning of cutoff stability, and in particular about the second condition above, we first give an alternative characterization using minimal cutoff scores.}

Let us define the notion of a matching induced by minimal cutoff scores, as explored by \citet{FleinerJanko2014} in a model without distributional constraints, and by \citet{BiroKiselgof2015} for a college admission model with ties. Let $d^{-p}$ denote the cutoff scores after decreasing the cutoff of $p$ by one, and keeping the other cutoffs the same, i.e., $d^{-p}(p)=d(p)-1$, and $d^{-p}(p')=d(p')$ for every $p'\neq p$. We say that cutoffs $d$ are \emph{minimal} if we cannot decrease the cutoff score of any project without making the induced matching infeasible. More formally for every project $p$, either $d(p)=0$ or the matching induced by $d^{-p}$, which we call $M^{-d}$, is not feasible.

\begin{proposition}\label{prop:minimal}
A matching is cutoff stable if and only if it is induced by minimal cutoff scores.
\end{proposition}

\begin{proof}
	
Due to Proposition \ref{prop:fair}, we only need to show that the minimality of the cutoff score is equivalent to cutoff non-wastefulness. Suppose first that matching $M$ is induced by minimal cutoff scores $d$, which means that for any project $p$, if $d(p)>0$ then $d^{-p}$ is not feasible.
Consider project $p$ with $d(p)>0$ and let $a$ be the applicant who is assigned a score of $d^{-p}(p)$ by project $p$. $M^{-p}$, the matching induced $d^{-p}$, is not feasible which means $M^{-p}\neq M$. Therefore, $M^{-p}=(M\cup \{(a, p)\})\setminus \{(a,M(a))\}$ and thus $p\succ_a M(a)$, which in turn means that $a$ is the highest ranked applicant who forms a blocking pair with $p$. Thus, the induced matching $M$ is cutoff non-wasteful.
Now suppose that $M$ is a cutoff stable matching and let $d$ be some cutoff scores that induce $M$ which are minimal in that sense that no smaller cutoff scores can induce $M$. Consider project $p$ with $d(p)>0$. The matching $M^{-p}$ induced by $d^{-p}$ must be different from $M$. Since $M$ is cutoff non-wasteful, $M^{-p}$ must be infeasible for every project $p$ with $d(p)>0$ and thus cutoffs $d$ are minimal.
\qed\end{proof}

\textcolor{black}{Going back to the interpretation of the second condition in the first definition of cutoff stability, we can see that it is possible that a blocking pair $(a, p)$ exists that would violate strong stability, but then there must be another blocking pair $(a',p)$, where $a'$ is higher in the preference of $p$ than $a$ and blocking with this pair would make the matching infeasible.}


We can show the natural correspondence in between the three solution concepts as follows.

\begin{proposition}
Every strongly stable matching is cutoff stable, and every cutoff stable matching is weakly stable.
\end{proposition}

\begin{proof}
By definition, strong non-wastefulness implies cutoff non-wastefulness, and cutoff non-wastefulness implies weak non-wastefulness.
\qed\end{proof}

Furthermore, these notions do not coincide, as illustrated in the following example.

\begin{example}[An instance where the sets of strongly stable, cutoff stable and weakly stable matchings are distinct]
	\begin{align*}
		A=\{a_1, a_2, a_3\}\\
		\succ_{a_1}:p_1, p_2, p_3&&
		\succ_{a_2}:p_2, p_1&&
		\succ_{a_3}:p_3\\
		P=\{p_1, p_2, p_3\} && c_{p_1}=c_{p_2}=c_{p_3}=1\\
		\succ_{p_1}:a_2, a_1&&
		\succ_{p_2}:a_1, a_2&&
		\succ_{p_3}:a_1, a_3 \\
		S=\{s\}&&
		P_{s}=\{p_1, p_2, p_3\}&&
		q_{s}=2 \\
	\end{align*}

Due to the supervisor capacity every feasible matching has size at most two, and any matching of size fewer that two is weakly wasteful. Thus, all the relevant matchings for weak, cutoff and strong stability have size two, so we consider all the matchings of size two below.
	$M_1=\{(a_1,p_1),(a_2,p_2)\}$ and $M_2=\{(a_1,p_2),(a_2,p_1)\}$ are strongly stable.
	
	$M_3=\{(a_1,p_2),(a_3,p_3)\}$ is not strongly stable, due to blocking pair $(a_1,p_1)$. However, $M_3$ is cutoff stable, since $a_2\succ_{p_1} a_1$ and $M_3\cup \{(a_2,p_1)\}$ is not feasible.
	
	$M_4=\{(a_1,p_3),(a_2,p_1)\}$ is not cutoff stable, due to blocking pairs $(a_1,p_2)$ and $(a_2,p_2)$. However, $M_4$ is weakly stable as adding an applicant to $a_2$ results in an infeasible matching. All other matchings of size two are not fair or not feasible.
\end{example}

Note that for the classical college admission model these notions are equivalent to stability as formulated by \citet{GaSh62}. Further theoretical findings about cutoff scores for this basic model are discussed by \citet{AzevedoLeshno2016}.


\vspace{-0.5em}
\subsection{Algorithm for computing a cutoff stable matching}

We present an algorithm (Algorithm~\ref{algo:KK}) that shows the existence of a cutoff stable matching for every instance of HER, and can compute a cutoff stable matching for an instance of SIP in strongly polynomial-time. 
Algorithm~\ref{algo:KK} has similarities with the algorithm proposed by \citet{KaKo17b} which finds a weakly stable matching for HER. In contrast to the algorithm of \citet{KaKo17b}, we do not explicitly work with a set of blocking of pairs, but with cutoff scores. Furthermore, a blocking pair $(a, p)$ was satisfied in \citet{KaKo17b} whenever the new matching $M\cup\{(a, p)\}$ stayed feasible, leading to a weakly stable solution in the end, whilst we satisfy a blocking pair $(a, p)$ if $(M\cup\{(a, p)\})\setminus \{(a,M(a))\}$ stay feasible, leading to a cutoff stable matching. Algorithm~\ref{algo:KK} can also be viewed as a modified version of the Fleiner-Jank\'o score-decreasing algorithm (see subsection 4.3 in \citet{FleinerJanko2014}). Whereas \citet{FleinerJanko2014} do not consider distributional constraints, our goal is to achieve stability properties under general heredity constraints.
Finally, in independent and recent work, \citet{KaKo20a} consider a fixed-point approach based on cutoffs. However, their heredity constraints apply to individual hospitals/projects and cannot capture SIP or regional constraints.

The idea of Algorithm~\ref{algo:KK} is simple. We start with maximum cutoffs that induce the empty matching and then we gradually decrease them until we can no longer do so without making the induced matching infeasible. Note that for the classical college admission problem this process is equivalent to the college proposing Gale-Shapley algorithm \citep{GaSh62}, that was observed in the Turkish college admission practice \citep{balinski1999tale}, and also in the Hungarian college admission scheme \citep{BiroKiselgof2015} in a more general form, as ties have been present in the rankings.

\begin{algorithm}		
	\begin{algorithmic}
	\REQUIRE lists $\succ_p $ for all $p\in P$ and $\succ_a$ for all $a\in A$; feasibility function $f$; project order $P^*=(p_1,...,p_k)$
	\ENSURE Matching $M$ and corresponding cutoffs $d_M$
	\end{algorithmic}
\begin{algorithmic}[1]
	\STATE Initialize $M$ to empty and $d_M(p)=|A|+1$ for every project $p$.
    \WHILE{Cutoff $d_M$ are not minimal}
    \STATE Locate the first $p_j$ in the list $P^*$ such that $M^{-p_j}$ is feasible.
    \STATE Let $M=M^{-p_j}$ and $d_M=d_M^{-p_j}$.
    \ENDWHILE 	
\end{algorithmic}
 \caption{Algorithm for matching with heredity constraints. }
 \label{algo:KK}
\end{algorithm}

\begin{theorem}\label{prop:weaklystable}
	A cutoff stable matching always exists for a matching problem under any set of distributional constraints that can be represented by a feasibility function, i.e., for the HER model. Algorithm~\ref{algo:KK} produces one such matching.
\end{theorem}
\begin{proof}
Note that $M$ remains feasible and also fair during the algorithm since it is induced by cutoffs. The final matching is cutoff stable, due to Proposition \ref{prop:minimal}, since the algorithm terminates when the cutoffs are minimal.
\qed\end{proof}

We remark here that the argument for Theorem~\ref{prop:weaklystable} does not require projects to be selected in the order $P^*$.


The next theorem shows that as long as the feasibility of a matching can be tested in polynomial time, Algorithm~\ref{algo:KK} runs in polynomial time.

\begin{theorem}\label{kk:runningtime}
	Suppose checking $f(w)$ takes $t$ time. Then, the running time of Algorithm~\ref{algo:KK} is $O(|A|\cdot|P|^2t)$.
\end{theorem}
\begin{proof}
We decrease the cutoffs in at most $(|A|+1)|P|$ rounds. In every round we potentially need to check the matching induced by decrementing the cutoff of each of $|P|$ projects. To do this for project $p_j$, we compute matching $M^{-p_j}$ from the current matching $M$, by simply comparing whether the newly admitted applicant $a_i$ with score $d_M^{-p_j}$ from $p_j$ prefers $p_j$ to her current match $M(a_i)$, and if so, whether the new matching $(M\cup\{(a_i,p_j)\})\setminus \{(a_i,M(a_i))\}$ is feasible. By assuming that checking the feasibility of a matching takes $t$ time, the overall run time of the algorithm is $O(|A||P|^2t)$.
\qed\end{proof}

We note that Algorithm~1 can be applied to the summer internship problem where the feasibility function is based on the supervisor budgets. Therefore, for the summer internship problem a cutoff stable matching exists. Furthermore, from Theorem~\ref{thm:feasibilitycheck} we know that the feasibility of a matching can be checked in polynomial time, and thus a cutoff stable matching can be found in polynomial time.


Additionally, using the cutoff-decreasing algorithm we can find a natural relationship between weakly stable and cutoff stable matchings.
\begin{proposition}
Every weakly stable matching that is not cutoff stable is Pareto-dominated by a cutoff stable matching for the applicants.
\end{proposition}

\begin{proof}
Let $M$ be a weakly stable matching that is not cutoff stable. Since $M$ is fair, it can be induced by cutoff scores, but these cutoff scores are not minimal, since $M$ is not cutoff stable. Therefore, we can gradually decrease these cutoff scores as described in the cutoff decreasing process until we reach a cutoff stable matching $M'$. Note that $M'$ Pareto-dominates $M$ since decreasing cutoff scores can only make the applicants better off, and every change in the matching means a strict improvement for at least one applicant.
\qed\end{proof}

Although Algorithm~\ref{algo:KK} satisfies cutoff stability, it also has some drawbacks.
Next we establish some properties of the algorithm.

\begin{theorem}\label{th:notgood}
	The following properties hold for Algorithm~\ref{algo:KK}.
	\begin{enumerate}
		\item Algorithm~\ref{algo:KK} is not strategyproof for the applicants.
		\item  Algorithm~\ref{algo:KK} does not always find a strongly stable matching whenever one exists.
		\item Changing the order of projects ordered after $p$ by $P^*$ can change $p$'s allocation.
		\item There exist cutoff stable matchings that cannot be produced as a result of Algorithm~\ref{algo:KK} by changing the project order $P^*$.
	\end{enumerate}
\end{theorem}


\begin{proof}
	We prove each of the statements separately.
	\begin{enumerate}
		\item 
	Consider the following instance.
		%
		%
		%
		%
		%
		%
		%
		%
		%
		%
			\begin{align*}
				A=\{a_1,a_2\}&&\succ_{a_1}:p_2,p_1&&\succ_{a_2}:p_2\\
				P=\{p_1,p_2\}&&\succ_{p_1}:a_2,a_1&&\succ_{p_2}:a_2,a_1\\
				S=\{s\}&&P_{s}=\{p_1,p_2\}&&
				q_{s}=c_{p_1}=c_{p_2}=1
				\end{align*}

			If we set $P^*=(p_1,p_2)$ then the algorithm outputs $M=\{(a_1,p_1)\}$. However, if $a_2$ were to modify their preferences to $\succ_{a_2}:p_2,p_1$, then the algorithm will output $M=\{(a_2,p_2)\}$, which is preferred by $a_2$. Thus, {the algorithm is not strategy-proof for applicants}.

		\item 	For the example above, since $(a_2,p_2)$ is the only strongly stable matching, the algorithm {does not find the strongly stable matching when one exists.}
		\item 	\textcolor{black}{Consider Example \ref{ex:multiple}. We can obtain each of the two cutoff stable matchings by changing the order the two projects are processed.}

		\item This is a consequence of the project-proposing nature of the algorithm. Consider the following instance.
	\begin{align*}
		A=\{a_1,a_2\}&&
		\succ_{a_1}:p_1,p_2&&
		\succ_{a_2}:p_2,p_1\\
		P=\{p_1,p_2\}&&
		\succ_{p_1}:a_2,a_1&&
		\succ_{p_2}:a_1,a_2\\
		S=\{s\}&&
		P_s=P\\
		c_{p_1}=c_{p_2}=1&&q_s=2
	\end{align*}
	Both possible project orders produce $M=\{(a_2,p_1),(a_1,p_2)\}$, but the matching $M=\{(a_1,p_1),(a_2,p_2)\}$ is also cutoff stable. Thus, {there exist weakly stable matchings that cannot be produced by the algorithm.}
	\end{enumerate}
	This completes the proof.
\qed\end{proof}

Since the examples above contain exactly one supervisor, Theorem~\ref{th:notgood} also applies to REG and any setting that contains REG as a special case.

\vspace{-0.5em}
\subsection{Finding a maximum size cutoff stable matching}

First, we observe that this problem is NP-hard for the REG model by part i) of the proof of Theorem \ref{NP-completeness}. Then we provide a mixed integer linear programming formulation.

\begin{theorem}
  Finding a maximum size cutoff\textcolor{black}{/weakly} stable matching is NP-hard for REG, \textcolor{black}{even if every hospital has one seat and each region has at most two hospitals}.
\end{theorem}

\begin{proof}
\textcolor{black}{
Part i) of the proof of Theorem \ref{NP-completeness} established a one-to-one correspondence in between the weakly stable matchings of an SMTI instance $I$ and the strongly/cutoff/weakly stable matchings of a SIP instance $I'$, where the sizes of the corresponding matchings are the same. Since finding a maximum size weakly stable matching for SMTI is NP-hard, so is the problem of finding a maximum size cutoff/weakly stable matching for SIP. Note that in this construction each project has capacity one and every supervisor has at most two distinct projects, which corresponds to the case of REG, where each hospital has capacity one and every region has at most two hospitals.}
\color{black}
\qed\end{proof}

Furthermore, since MAX-SMTI is not approximable within a factor of 21/19 unless $P=NP$ \citep{halldorsson2007improved}, due to the one-to-one relation of the sizes of the matchings in our reduction, the same inapproximability result applies for our problems as well.
\vspace{-0.5em}
\subsection*{MILP-formulation}

The mixed integer linear program will have three main parts. The first set of constraints describes the feasibility of the funding allocation. The matching in between applicants and projects are described with (0-1) binary variables defined for every mutually acceptable applicant-project pair, as follows, let $y_{a, p}=1\iff (a, p)\in M$. Let $x_{s,p}$ be a non-negative continuous variable denoting the funding for project $p$ provided by supervisor $s$. The following three sets of conditions ensure the feasibility of the solution.
\begin{equation}
\label{eq:applicant_feasible}
\text{Applicant-feasibility:  }\sum_{p\in P}y_{a, p}\leq 1 \mbox{ for each } a\in A
\end{equation}

\begin{equation}
\label{eq:project_feasible}
\text{
Project-feasibility:  }\sum_{a\in A}y_{a, p}=\sum_{s\in S}x_{s,p}\leq c_p \mbox{ for each } p\in P
\end{equation}

\begin{equation}
\label{eq:supervisor_feasible}
\text{\textcolor{black}{Budget}-feasibility: }\sum_{p\in P}x_{s,p}\leq q_s \mbox{ for each } s\in S
\end{equation}

The second set of constraints will describe the fairness of the matching by means of cutoff scores. Let $z_{a, p}$ be the score of applicant $a$ at project $p$, a given integer constant in the interval $[0,|A|]$. For every project $p$, let $d(p)$ be an integer variable in the range of $[0,|A|+1]$ denoting the cutoff of project $p$. We can link the cutoff scores with the induced matching $M$ by the following set of constraints, as also used in \citet{ABMcB2016} and \citet{Delormeetal2019}.

\begin{equation}
\label{eq:cutoff1}
d(p)\leq (1-y_{a, p})(|A|+1) + z_{a, p} \mbox{ for each } (a, p)
\end{equation}

The above constraint enforces that if an applicant is assigned to a project then she reached the cutoff there.

\begin{equation}
\label{eq:cutoff2}
z_{a, p}+1\leq d(p) + \left(\sum_{p'\succeq_a p} y_{a, p'}\right)\cdot (|A|+1) \mbox{ for each } (a, p)
\end{equation}

This constraint implies that if the applicant is rejected from project $p$, so she is not admitted there and nor to any better project of her preference (in which case the sum term is zero on the right hand side), then the cutoff must be higher at $p$ than her score there. These two sets of conditions together imply that every student is admitted to her most preferred project where she achieved the cutoff, which means that we get the matching induced by the cutoffs.

For ensuring cutoff stability we have to provide cutoff non-wastefulness by enforcing that the cutoffs are minimal. This can be achieved by simply minimizing the sum of the cutoff scores in the objective function of the MILP. To see this we just need to observe that if the matching would not be cutoff minimal, so we could decrease some of the cutoffs by keeping the solution feasible, but then the solution was not optimal with respect to this objective function.

Finally, if we want to maximize the number of applicants matched in the cutoff stable solution then we can use the following combined objective function, where $W$ is a large enough constant, \textcolor{black}{namely, $W > (|A|+1)\cdot |P|$}.

\begin{equation}
\label{eq:objective}
\max \sum_{a\in A, p\in P} W\cdot y_{a, p} - \sum_{p\in P}d(p)
\end{equation}

This completes the MILP for finding a maximum size cutoff stable matching.

\section{Egalitarian Budget Allocations}

Given a feasible matching $M$, there can exist multiple ways to allocate supervisor budgets among projects to fund all applicants matched to them. Our goal is to find a method for the fairest such allocation. We have chosen to deal with fairness post-match, in order to find a solution that does not constrain the set of feasible matchings.

\begin{algorithm}[h!]
	\color{black}
	\caption{Computing an egalitarian funding allocation for a given feasible matching}
	\label{algo:Fairness}
	\begin{algorithmic}

		\REQUIRE Feasible matching $M=\{(s,p)|s\in S, p\in P_s\}$; Supervisor budget quantities $q_s$ $\forall s\in S$; Targets $\{t_{s,p}|(s,p)\in M\}$.
		\ENSURE Funding allocation $x$
	\end{algorithmic}
	\begin{algorithmic}[1]
	
		\STATE $M_{tight}\leftarrow \emptyset$
		\STATE $\{x_{s,p}|(s,p)\in M\}\leftarrow$ Feasible funding allocation for matching M (Theorem~\ref{thm:feasibilitycheck})
		\WHILE{$M_{tight} \neq M$}
		\STATE Minimise the maximal component of $\Big\{\frac{x_{s,p}}{t_{s,p}}| {(s,p)\in M} \Big\}$,  that is not yet tight, allowing $x_{s,p}$ to vary but holding all $t_{s,p}$ fixed, by solving the following LP:
			\begin{align*}
			\lambda^* \quad \leftarrow \text{Min } \lambda \text{ s.t. }\\
			 x_{s,p}\geq 0
			&\quad \forall (s,p) \in M  &\\
			\sum_{s\in S_p} x_{s,p} = |\{(s,p)\in M|s\in S_p\}|
			&\quad \forall p\in P &\\
			\sum_{p\in P_s} x_{s,p} \leq q_{s}
			&\quad \forall s\in S &\\
			\frac{x_{s,p}}{t_{s,p}} \leq \lambda
			&\quad \forall (s,p)\in M\backslash M_{tight}   &\\
			\frac{x_{s,p}}{t_{s,p}} = \lambda_{s,p}
			&\quad \forall (s,p)\in M_{tight} &
			\end{align*}
		
		\FOR{$(s,p)^*\in M\backslash M_{tight}$}
		\STATE Determine whether the constraint corresponding to $(s,p)^*$ is tight by solving the following Auxiliary LP:
		\begin{align*}
			\epsilon^*  \leftarrow   \text{Max }  \epsilon &\quad \text { s.t. }\\
x_{s,p}\geq 0  &\quad \forall (s,p) \in M  \\
\sum_{s\in S_p} x_{s,p} = {|\{(s,p)\in M|s\in S_p\}|} &\quad \forall p\in P \\
\sum_{p\in P_s} x_{s,p} \leq q_{s} &\quad \forall s\in S\\
\frac{x_{s,p}}{t_{s,p}} \leq \lambda^* &\quad \forall (s,p)\in M\backslash M_{tight}\backslash \{(s,p)^*\}\\
\frac{x_{s,p}}{t_{s,p}} = \lambda_{s,p} &\quad \forall (s,p)\in M_{tight} \\
\frac{x_{s,p}}{t_{s,p}} + \epsilon \leq \lambda^* &\quad \text{ for } (s,p)=(s,p)^*
			\end{align*}
		\IF{$\epsilon^*=0$}
		\STATE Add $(s,p)^*$ to $M_{tight}$
		\STATE $\lambda_{s,p}\leftarrow \lambda^*$
		\ENDIF
		\ENDFOR
		\ENDWHILE
		\STATE $\{x^*_{s,p}|(s,p)\in M\}\leftarrow$ Funding allocation that solves the LP in step 4 with $\{\lambda_{s,p}|(s,p)\in M_{tight}=M\}$
		\RETURN $\{x^*_{s,p}|(s,p)\in M\}$
	\end{algorithmic}
\end{algorithm}

For each supervisor $s\in S$, and each project they supervise $p\in P_s$ set $0< t_{s,p}\leq 1$ such that, for each $p\in P$, $\sum_{s\in S_p} t_{s,p}=1$. Call this the normative `target' - how much we would want $s$ to contribute to the funding of any applicant matched to $p$.
For instance, if we would ideally want the supervisors of any given project to contribute equally to its funding, we would set:
$t_{s,p}=\frac{|M(p)|}{|S_p|} \ \ \forall p\in P, s\in S_p.$
However, given a feasible matching $M$, some supervisors may lack sufficient funding to reach these targets. Therefore, we seek to find a funding allocation that is closest to the target allocations. In order to define closest, we consider specific lexicographic comparisons.
Let $X=\{x_{s,p}\}_{s\in S,p\in P_s}$ be a feasible funding allocation for matching $M$. Denote by $\phi_X$ the vector corresponding to the weakly decreasing ordering of the set:
$\Big\{\frac{x_{s,p}}{t_{s,p}}\Big\}_{s\in S,p\in P_s}.$
Denote by $\Phi_M$ the set of such vectors corresponding to all feasible funding allocations for matching $M$.

The \textbf{egalitarian feasible funding allocation} for matching M is the funding allocation corresponding to $\phi^* \in \Phi_M$ such that for all $\phi \in \Phi_M$, $\phi^*\prec_{lex} \phi$, where $\prec_{lex}$ refers to the well-known lexicographic order. This is exactly equivalent to finding the leximin optimum of the following set: $\Big\{\frac{-x_{s,p}}{t_{s,p}}\Big\}_{s\in S,p\in P_s}.$
%
The egalitarian feasible funding allocation can be achieved via Algorithm~\ref{algo:Fairness}, which runs a series of linear programs.

\begin{theorem}\label{fairness:runningtime}
	Given a feasible matching, the egalitarian funding allocation can be computed in time $O(|S|^2|P|^2)LP(O(|S|\cdot|P|))$, where LP refers to the running time of the linear programming algorithm used.
\end{theorem}
\begin{proof}
	The algorithm solves a series of linear programs with $O(|S|\cdot|P|)$ constraints. The while loop iterates at most $|M|\leq|S||P|$ times, as the algorithm adds at least one element to the set $M_{tight}$ in every iteration. The for loop also iterates at most $|M|$ times. Thus, the overall complexity is $O(|S|^2|P|^2)LP(O(|S|\cdot|P|))$, where LP refers to the running time of the linear programming algorithm used. \textcolor{black}{Finding a feasible funding allocation for matching M at the beginning of the algorithm does not increase its complexity}. Since linear programs can be solved in polynomial time, the egalitarian funding allocation can also be computed in polynomial time. \textcolor{black}{It is possible that this could be improved to strongly polynomial time by reducing to the fair integral flow problem \citep{FrMu2021}. We leave the existence of such a reduction as an open question.}
\qed\end{proof}

\vspace{-1em}
\section{Conclusion}

We presented a novel matching model that captures many real-world scenarios.
For the model, we presented a compelling solution that is polynomial-time and satisfies stability and fairness properties. Our central algorithm computes a cutoff stable matching even for the general matching model with heredity constraints. Hence, it applies to many applications such as refugee matching that involve heredity constraints.

Several directions and problems arise as a result of our study.
Our approach to finding a fair budget allocation was to first compute a cutoff stable matching and then find an egalitarian budget allocation. It will be interesting to explore a fair outcome that is fairest in some global sense across all weakly stable matchings.
It had been open whether there exists an algorithm that is strategyproof and satisfies weak stability or cutoff stability. Just recently, \citet{CKM+22a} proved that strategyproofness and weak stability are incompatible in general. Understanding the conditions under which strategyproofness and stability concepts are compatible remains an interesting direction.

\vspace{-1em}
\section*{Acknowledgments}

Aziz gratefully acknowledges support from Defence Science and Technology (DST). Baychkov's work was supported by the CSIRO undergraduate vacation scholarship program, and he extends many thanks to Haris Aziz and Gavin Walker for their invaluable mentorship. Bir\'o is supported by the Hungarian Academy of Sciences, Momentum Grant No. LP2021-2, and the Hungarian Scientific Research Fund, OTKA, Grant No.\ K128611.

\small
\vspace{-1em}
\bibliographystyle{plainnat}

\end{document}